\documentclass[11pt]{article}

\usepackage{times}
\usepackage{amsmath,amssymb,amsthm}
\usepackage{xspace}
\usepackage{expdlist}
\usepackage{authblk}
\usepackage{framed,varwidth}
\usepackage[]{algpseudocode}
\usepackage[margin=1in]{geometry}
\usepackage{xspace}
\usepackage{thmtools,thm-restate}
\usepackage{wasysym}
\usepackage{enumitem}
\usepackage{nicefrac}
\usepackage{graphicx}

\usepackage[font=small,skip=0pt]{caption}

\setlist[description]{leftmargin=1cm,labelindent=1cm}

\def\polylog{\operatorname{polylog}}

\newtheorem{theorem}{Theorem}[section]

\newtheorem{lemma}[theorem]{Lemma}
\newtheorem{fact}[theorem]{Fact}
\newtheorem{definition}[theorem]{Definition}

\newtheorem{example}[theorem]{Example}

\newcommand{\Mh}{\mathcal{M}}
\newcommand{\Ih}{\mathcal{I}_i}
\newcommand{\Ihj}{\mathcal{I}_{i,j}}

\newcommand{\eps}{\epsilon}
\newcommand{\mui}{\mu_{i}}
\newcommand{\muij}{\mu_{i,j}}
\newcommand{\wmuij}{\widetilde{\mu}_{i,j}}
\newcommand{\wmui}{\widetilde{\mu}_{i}}

\newcommand{\runs}{\text{runs}}
\newcommand{\allruns}{\ensuremath{\text{runs}_\ell}}
\newcommand{\query}{\ensuremath{\textsc{Diff}}\xspace}
\newcommand{\update}{\ensuremath{\textsc{NewRun}}\xspace}
\newcommand{\Ts}{\ensuremath{T^\star}\xspace}
\newcommand{\Hs}[3]{\ensuremath{\textsc{Ham}(#1,#2)[#3]\xspace}} 
\newcommand{\Ham}{\ensuremath{\textsc{Ham}}\xspace}
\newcommand{\HD}{\ensuremath{\Delta}\xspace}

\newcommand{\rleAlg}{\ensuremath{\mathcal{A}_{\scalebox{0.5}{\text{RLE}}}}\xspace}
\newcommand{\approxAlg}{\ensuremath{\mathcal{A}_{\scalebox{0.5}{\text{Approx}}}}\xspace}

\begin{document}

\title{The $k$-mismatch problem revisited}

\author[1]{Rapha\"el Clifford}
\author[1]{Allyx  Fontaine}
\author[2]{Ely Porat}
\author[1]{\\Benjamin Sach}
\author[1]{Tatiana Starikovskaya}

\affil[1]{University of Bristol, Department of Computer Science, Bristol, U.K.}
\affil[2]{Bar-Ilan University, Department of Computer Science, Israel}

\date{\empty}

\maketitle

\begin{abstract}
We revisit the complexity of one of the most basic problems in pattern matching. In the $k$-mismatch problem we must compute the Hamming distance between a pattern of length $m$ and every $m$-length substring of a text of length $n$, as long as that Hamming distance is at most $k$. Where the Hamming distance is greater than $k$ at some alignment of the pattern and text, we simply output ``No''. 

We study this problem in both the standard offline setting and also as a streaming problem.  In the streaming $k$-mismatch problem the text arrives one symbol at a time and we must give an output before processing any future symbols. Our main results are as follows:
\begin{itemize}
\item Our first result is a deterministic $O(nk^2\log{k}/m+n\polylog{m})$ time \emph{offline} algorithm for $k$-mismatch on a text of length $n$. This is a factor of $k$ improvement over the fastest previous result of this form from SODA 2000~\cite{ALP:2000,ALP:2004}. 
\item We then give a randomised and online algorithm which runs in the same time complexity but requires only $O(k^2\polylog{m})$ space in total. 
\item Next we give a randomised $(1+\epsilon)$-approximation algorithm for the streaming $k$-mismatch problem which uses $O(k^2\polylog{m}/\epsilon^2)$ space and runs in $O(\polylog{m}/\epsilon^2)$ worst-case time per arriving symbol.
\item Finally we combine our new results to derive a randomised $O(k^2\polylog{m})$ space algorithm for the streaming $k$-mismatch problem which runs in  $O(\sqrt{k}\log{k} + \polylog{m})$ worst-case time per arriving symbol. This improves the best previous space complexity for streaming $k$-mismatch from FOCS 2009~\cite{Porat:09} by a factor of $k$. We also improve the time complexity of this previous result by an even greater factor to match the fastest known offline algorithm (up to logarithmic factors).
\end{itemize} 
\end{abstract}

\section{Introduction\label{sec:intro}}

We study the complexity of one of the most basic problems in pattern matching. In the $k$-mismatch problem we are given as input two strings, a pattern of length $m$ and a text of length $n$. The task is to output the Hamming distance between the pattern and every $m$-length substring of the text where the Hamming distance is at most $k$. If the Hamming distance is greater than $k$ we need only output ``No''.  We provide new, faster and more space efficient solutions for the $k$-mismatch problem in both the classic offline setting and when considered as an online streaming problem.

The general task of efficiently computing the Hamming distances between a pattern and a longer text has been studied since at least the 1980s when  $O(n\sqrt{m\log{m}})$ time solutions were first discovered~\cite{Abrahamson:1987,Kosaraju:1987}.  For many years however the fastest known algorithm for the $k$-mismatch problem ran in $O(nk)$ time~\cite{LV:1986a} using repeated Lowest Common Ancestor calls to a generalised suffix tree of the pattern and text.  Eventually, in the year 2000 two improved algorithms were given which run in  $O(nk^3\log{k}/m+n)$ and $O(n\sqrt{k\log{k}})$ time respectively~\cite{ALP:2000,ALP:2004}.  The former algorithm is clearly preferable when $k/m$ is relatively small and the latter algorithm has superior performance in all other cases. Until this point, these two algorithms remain the fastest solutions known.

Our first result is a new deterministic algorithm for the $k$-mismatch problem which is faster than all previous solutions when $k  \in  O(m^{2/3-\epsilon})$.  This is a result of independent interest, providing the fastest known $k$-mismatch algorithm for a large and particularly natural range of values of the threshold $k$. 


\begin{theorem}\label{thm:deterministic}
 Given a pattern $P$ of length $m$ and a text $T$ of length $n$, there is a deterministic solution for the $k$-mismatch problem with run-time $O(nk^2\log{k}/m+n\polylog{m})$. 
\end{theorem}

We then turn our attention to a small-space online version of the $k$-mismatch problem.   In this setting the text arrives one symbol at a time and we must output the Hamming distance, if it is at most $k$, before the subsequent symbol arrives. We consider a particularly strong space model where we account for all the space used by our algorithm and in particular we are not permitted to store a copy of the pattern or text without also accounting for that.  We obtain the following result.

 \begin{theorem}\label{thm:amortised}
   Given a pattern $P$ of length $m$ and a streaming text of total length $n$ arriving one symbol at a time, there is a randomised $O(k^2\polylog{m})$ space online algorithm which runs in $O(nk^2\log{k}/m+n\polylog{m})$ time and solves the $k$-mismatch problem. The probability of error is at most~$1/m^2$.
 \end{theorem}

A particularly attractive feature of this new online algorithm is that whenever $k \in O(m^{1/2-\epsilon})$, it not only uses sublinear space but also has  total running time of only $O(n\polylog{m})$ time.

 We next consider a small-space  approximate version of the $k$-mismatch problem.  In return for tolerating a constant multiplicative error in the output we are able to give an algorithm that runs in $\polylog{m}$ time per symbol.     We define the $(1+\epsilon)$-approximate $k$-mismatch problem as
follows.  Let $y$ be the true Hamming distance at a particular alignment of the pattern and text. At each alignment of the pattern and text, we output either an integer $x$ or ``No''.  If we output ``No'' then $y>k$ with high probability.  If we output an integer $x$ then $y\leq x \leq  (1+\epsilon)y$ with high
probability. One subtlety with this problem definition is that the two cases overlap when $k < y \leq (1+\epsilon)k$. In this case we are free to either output ``No'' or an integer $x$. However any integer we do output must still be an $(1+\epsilon)$-approximation to the true Hamming distance.   This formulation is a generalisation of the $\epsilon$-threshold decision problem introduced by Indyk in FOCS 1998~\cite{Indyk:1998} where a linear space $O((n/{{\epsilon}^3}) \log{m})$ time offline algorithm was given. 


  \begin{theorem}\label{thm:approximate}
  Given a pattern $P$ of length $m$ and a streaming text arriving one symbol at a time, there is a randomised $O(k^2\polylog{m}/\epsilon^2)$ space algorithm which takes  $O(\polylog{m}/\epsilon^2)$ worst-case time per arriving symbol and solves the $(1+\epsilon)$-approximate $k$-mismatch problem. The probability of error is at most~$1/m^2$.
 \end{theorem}

Finally we turn to the streaming $k$-mismatch problem itself. Here the text arrives one symbol at a time, as in the  online model. However    a particularly important additional feature is that the performance per arriving symbol should be guaranteed worst-case.   The analysis of small space streaming algorithms for pattern matching problems started in earnest in FOCS 2009~\cite{Porat:09}.  In that year Porat and Porat presented a randomised algorithm for performing exact matching in a stream which only stored $O(\log{m})$ words of space and required $O(\log{m})$ worse-case time per arriving symbol~\cite{Porat:09}.  This result was subsequently slightly simplified~\cite{EJS:2010} and then eventually improved to take constant time per arriving symbol in 2011~\cite{BG:2011}.  

Following this early breakthrough, the natural question was to ask for what other pattern matching problems is it also possible to find near optimal time and space solutions. Unfortunately, it turns out that for a large range of the most popular pattern matching problems, including pattern matching with wildcards, $L_1$, $L_2$, $L_{\infty}$-distance and edit distance, space proportional to the pattern length is required for any randomised online algorithm~\cite{CJPS:2012}.  Despite this, the Porat and Porat paper also presented an algorithm for the streaming $k$-mismatch problem that ran in $O(k^3\polylog{m})$ space and $O(k^2\polylog{m})$ time per arriving symbol in their original 2009 paper. For small $k$ this is a sublinear space algorithm and it remains to date one of the few fast sublinear space algorithms for streaming pattern matching that is known.  

As our final result we use a combination of Theorems~\ref{thm:amortised} and~\ref{thm:approximate} as the basis for a new worst-case time streaming algorithm for the $k$-mismatch problem which is not only significantly faster than the result  of Porat and Porat, but whose time complexity matches (up to logarithmic factors) the fastest known offline algorithm.   Our method also uses a multiplicative factor of $k$ less space than the previous result of Porat and Porat  (up to logarithmic factors again) while still guaranteeing that an output is made after each arriving symbol and before any future symbol is processed.

\begin{theorem}\label{thm:streaming}
 Given a pattern of length $m$ and a streaming text arriving one symbol at a time, there is a randomised $O(k^2\polylog{m})$ space algorithm which takes $O(\sqrt{k}\log{k} + \polylog{m})$  worst-case time per arriving symbol and solves the $k$-mismatch problem. The  probability of error is at most $1/m^2$.
\end{theorem}

Each one of our four main results is of independent interest and advances the state of the art for their respective problems. However, we regard Theorems~\ref{thm:deterministic} and~\ref{thm:streaming} to be the most significant contributions of this paper. The main technical contributions are set out in Section~\ref{sec:overview}.

 \section{Related work and lower bounds}

There has been great interest in time and space efficient streaming algorithms over the last 20 years, following the seminal work of~\cite{AMS:1996}.  In relation specifically to pattern matching problems,  where space is not limited but where an output must be computed after every new symbol of the text arrives, the Hamming distance between the pattern and the latest suffix of the stream can be computed online in $O(\sqrt{m\log{m}})$ worst-case time per arriving symbol or $O(\sqrt{k}\log{k} + \log{m})$ time for the $k$-mismatch version~\cite{CS:2010}.  Both these methods however require $\Theta(m)$ space.  Using the same approach, a number of other approximate pattern matching algorithms have also been transformed into efficient linear space online algorithms including~\cite{AAKLP:2008,AABLLPSV:2009,AALP:2006,AFM:1994, AEE:2006,ACHP:2003,LV:1988kdiff}. The only other small space streaming pattern matching algorithm that we are aware of solves a problem known as parameterised matching~\cite{JPS:2013}.  In the offline setting, randomised and deterministic algorithms that give an $(1+\epsilon)$-approximation to the Hamming distance are also known~\cite{Karloff:1993}. The running time of these two algorithms is $O((n/{\epsilon}^2)\log^2{m})$  and $O((n/{\epsilon}^2)\log^3{m})$ respectively.  Using an existing online to offline reduction~\cite{CEPP:2011} the $(1+\epsilon)$-approximation algorithms of~\cite{Karloff:1993} can be converted into $\Theta(m/\eps^2)$ space online solutions with guaranteed worst case running time per arriving symbol at a multiplicative time cost of $O(\log{m})$.
 
 
One can derive a space lower bound for any streaming problem by looking at a related one-way communication complexity problem. The randomised one-way communication complexity of determining if the Hamming distance between two $n$ bits strings is greater than $k$ is known to be $\Omega(k)$ bits (with an upper bound of $O(k\log{k})$~\cite{HSZZ:06}. From this we can derive the same lower bound for the space required by any streaming $k$-mismatch algorithm.  The results we present in this paper take us a significant step towards this lower bound but it is still unclear how closely it can ultimately be reached.

\section{Overview of the main ideas}\label{sec:overview}
In this section we will give an overview of the main ideas needed to prove Theorems~\ref{thm:deterministic},~\ref{thm:amortised},~\ref{thm:approximate} and~\ref{thm:streaming}.   

We start by introducing the notion of the approximate period, or $x$-period of a string. This idea will be crucial for all of our main results. We will in general use the approximate period of the pattern to separate our problems into two cases.  Let $\Ham (P, S)$ be the Hamming distance between equal length strings $P$ and $S$ and let $\Hs{P}{T}{i}$ be $\Ham(P, T[i-m+1,i])$.

\begin{definition}\label{def:approxperiod}
The $x$-period of a string $P$ of length $m$ is the smallest integer $\pi > 0$ such that $\Ham(P [\pi, m-1], P [0, m - 1-\pi]) \le x$. (For example, the $1$-period of a string $babaa$ is $2$.)
\end{definition}

Let $\ell$ be the $3k$-period of the pattern $P$  and  as our first of two cases, consider when $\ell \leq k$. We call this the small approximate period case and as we will see, the solution for this case contains some of the main ideas on which our other results will rely.

\begin{fact}\label{fact:x-period}
 If a pattern has $3k$-period $\ell$ then each $(3k/2)$-mismatch of the pattern and the text must be at least $\ell$ symbols apart.
\end{fact}

\paragraph{Small approximate period ($\ell \leq k$) case of Theorems~\ref{thm:deterministic} and~\ref{thm:amortised}.} Our solution for the small approximate period case is the same for both our offline (see Theorem~\ref{thm:deterministic}) and online small-space (see Theorem~\ref{thm:amortised}) algorithms. The main new idea is to reduce the problem to many instances of run length encoded pattern matching.  Our solution utilises a simple variant of run length encoding and we will use this encoding to reduce the $k$-mismatch problem to a total of $O(k^2)$ small instances of the run length encoded Hamming distance problem.

There are a number of surprising elements to our solution. The first one is that in any substring of the text of length $2m$ we can find a compressible region that contains all the alignments of the pattern and text with Hamming distance at most $k$. The second is that by choosing a suitable partitioning of the pattern and of this compressible region into $O(k)$ subpatterns and $O(k)$ subtexts respectively and then run length encoding those, we can ensure that the total number of runs, summed across all subpatterns and subtexts is only $O(k)$. The third is that despite there being $O(k)$ subpatterns and $O(k)$ subtexts giving $O(k^2)$ instances of the run length encoded Hamming distance problem, each of which can take $O(k^2 \log{k})$ time, we show that the time complexity of all the instances sums to only $O(k^2 \log k)$. By the same approach, we will demonstrate that the working space of all the instances sums to $O(k^2)$. We will also need to be careful when recovering the final Hamming distances because, in the worst case, each final distance is the sum of $k$ outputs of the run length encoded Hamming distance problem. A naive summation would therefore result in an additive $\Omega(k)$ term per Hamming distance. To overcome this bottleneck we will take advantage of the compressed output to reduce the time taken to recover the final distances to $O(m + k^2\log k)$ per substring.  

Using a standard trick we run our algorithm independently on $O(n/m)$ substrings of the text of length $2m$, each overlapping the next by $m$ symbols, thus giving Lemma~\ref{lemma:smallperiod}. The main steps are set out in Algorithm~\ref{alg:det-smallperiod} with additional details and a proof overview set out in Section~\ref{sec:smallperiod}.  

\begin{figure}[ht]
\begin{center}
  \fbox{\quad%
  \begin{varwidth}{\linewidth}
 
\textbf{Input:} Pattern of length $m$ and text of length $2m$.
\begin{enumerate}
 \item Identify a compressible region of the text which contains all the $k$-mismatches.\label{step:identify}
 \item Partition this region into $O(k)$ subtexts and the pattern into $O(k)$ subpatterns.\label{step:partition}
 \item Run length encode all the subpatterns and subtexts.\label{step:RLE}
 \item Compute run length encoded Hamming distances for each subpattern/subtext pair.\label{step:RLE-Hamming}
 \item Sum the Hamming distances from Step~\ref{step:RLE-Hamming}.
\end{enumerate}
\end{varwidth}%
 \quad}
\end{center}
\caption{Deterministic algorithm for $k$-mismatch when the pattern has small approximate period.}\label{alg:det-smallperiod}
\end{figure}

\begin{lemma}\label{lemma:smallperiod}
Consider a pattern $P$ of length $m$, and a text $T$ of length $n$ arriving online. If the $3k$-period of $P$ is smaller than $k$, then the $k$-mismatch pattern matching problem can be solved in $O(k^2)$ space and $O(nk^2\log{k}/m +n)$ time. 
\end{lemma}

\paragraph{Large approximate period ($\ell > k$) case of Theorems~\ref{thm:deterministic} and~\ref{thm:amortised}.} The overall structure of our solutions for both Theorems~\ref{thm:deterministic} and~\ref{thm:amortised} when the pattern has large approximate period is the same. We first describe the simpler deterministic case which gives us Theorem~\ref{thm:deterministic}.

\begin{enumerate}
 \item Filter out all alignments of the pattern and text with Hamming distance greater than $3k/2$. We can do this by running  Karloff's $(1+\epsilon)$-approximation algorithm~\cite{Karloff:1993} with $\epsilon=1/2$, excluding all positions which are reported to have Hamming distance greater than $3k/2$. This takes $O(\log^3{m})$ time per symbol in the text. 
 \item Verify whether the Hamming distance is at most $k$ at those positions. This takes $O(k)$ time per alignment we need to verify using $O(k)$ repeated application of constant time longest common prefix (LCP) queries between the pattern and the suffix of the text starting at the current alignment~\cite{LV:1986a}.
\end{enumerate}

We need only run the verification step at alignments that have not been filtered out by the filtering step.  By Fact~\ref{fact:x-period} there can be no more than one such alignment for every $k$ consecutive text symbols that arrive. It follows that the total amortised time for the large approximate period case is $O(n\polylog{m})$.  This completes the algorithmic description that establishes Theorem~\ref{thm:deterministic}.

In order to establish Theorem~\ref{thm:amortised} for the large approximate period case we will need small-space versions of both the filtering and verification steps.  For the filtering step we set $\epsilon=1/2$ again and this time use Theorem~\ref{thm:approximate}, which we discuss later. In the same way as in the deterministic case, after filtering the verification step will only need to verify at most one potential $k$-mismatch per $k$ consecutive text symbols. To do this efficiently we maintain a dynamic data structure that allows us to query the Hamming distance between $P$ and the latest $m$-length suffix of the text and will output the exact distance if it is at most $k$ and  ``No'' otherwise.   Each time a new symbol of the text arrives we perform an update.

\begin{lemma}\label{lemma:validation}
 For a given pattern $P$ of length $m$, and an online text $T$ of length $n$ there is a data structure which answers Hamming distance queries as described above and uses $O (k^2 \polylog{m})$ space, update time  $O(\polylog{m})$, and query time $O(k \polylog{m})$. If the Hamming distance does not exceed $2k$, the probability of  error is at most $1/m^2$.
 \end{lemma}

The key technical innovation, which is set out in Lemma~\ref{lemma:validation} is that our data structure takes only $\polylog{m}$ time to perform an update when a new text symbol arrives if no query is performed at that time.  We will use this asymmetry in query and update times combined with Fact~\ref{fact:x-period}  to show Theorem~\ref{thm:streaming}.

Our solution for Lemma~\ref{lemma:validation} works by first reducing the problem to repeated application of $1$-mismatch, in a similar fashion to Porat and Porat~\cite{Porat:09} and then in turn reducing the $1$-mismatch problem to the streaming dictionary matching problem. However, our method differs significantly in technique from the previous work both by randomising the first reduction step and then in our second reduction step which allows us to perform updates much more quickly than queries.

\paragraph{$(1+\epsilon)$-approximate $k$-mismatch - Theorem~\ref{thm:approximate}.} The main new ideas for our approximation algorithm are a novel randomised length reduction scheme and a two stage approximation scheme. The general idea is as follows. First, during preprocessing we reduce the length of the pattern to be only $O(k \log^2{m})$. We then overcome a particularly significant technical hurdle by showing how to transform the text in such a way that any Hamming distance between the reduced length pattern and transformed text provides a reasonable approximation of the corresponding Hamming distance in the original input. Finally we apply an existing linear space online $(1+\epsilon)$-approximation algorithm to the reduced length pattern and the transformed text to give the final approximate answer. The entire process is repeated independently in parallel a logarithmic number of times to improve the error probability.  We argue that this approximation of an approximation still gives us a $(1+\epsilon)$-approximation to the true Hamming distance at each alignment with good probability.

\paragraph{Deamortisation using the tail trick - Theorem~\ref{thm:streaming}.} We can now describe how to deamortise our online $k$-mismatch algorithm with $O(nk^2\log{k}/m+n\polylog{m})$ run-time that we gave for Theorem~\ref{thm:amortised} to give us a fast worst-case time streaming algorithm satisfying Theorem~\ref{thm:streaming}.  We first observe that if the pattern length $m$ is at most $2k^2$, we can run an existing algorithm~\cite{CS:2010} which will take $O(\sqrt{k} \log k)$ time per symbol and uses linear space,  which in this case is $O(k^2)$. We now proceed under the assumption that $m > 2k^2$.

To deamortise the algorithm, we use a two part partitioning that we call the \emph{tail trick}. Similar ideas were also used to deamortise streaming pattern matching algorithms in~\cite{CAPSS:2015,CS:2010}. We partition the pattern into two parts: the \emph{tail}, $P_t$~--- the suffix of $P$ of length $2k^2$, and the \emph{head}, $P_h$~--- the prefix of $P$ length $(m-2k^2)$ . We will compute the current Hamming distance, $\Hs{P}{T}{i}$ by summing $\Hs{P_t}{T}{i}$ and $\Hs{P_h}{T}{i-2k^2}$. To compute $\Hs{P_t}{T}{i}$ we again use the existing linear space online $k$-mismatch algorithm from \cite{CS:2010} taking $O(\sqrt{k} \log k)$ time per symbol and $O(k^2)$ space.

We also need to make sure that when the $i$-th symbol of the text, $T[i]$, arrives, we will have computed $\Hs{P_h}{T}{i-2k^2}$  in time. To this end we run the amortised algorithm from Theorem~\ref{thm:amortised} using pattern $P_h$. However, we cap the run-time at $O (\polylog m)$ per symbol. That is, when $T[i]$ arrives we run $\polylog m$ steps of the algorithm. Because the algorithm is amortised, it may lag behind the text stream~--- when $T[i]$ arrives, it may still be processing $T[i']$ for some $i'<i$. Fortunately, the lag cannot exceed $2k^2$, that is at all times $i-i'\leq 2k^2$. This is because we are able to show that while processing any $k^2$ consecutive text symbols the total time complexity of the algorithm, summed over those consecutive symbols is upper bounded by $O(k^2 \log k) = O(k^2 \polylog m)$. 
To allow for the lag in the deamortisation process we also maintain a buffer containing the most recently arrived $2k^2$ text symbols and the most recent $2k^2$ outputs. 

The space is dominated by the algorithm from Theorem~\ref{thm:amortised}  which uses $O(k^2 \polylog m)$ space. The time complexity is the sum of the complexities for processing $P_t$ and $P_h$ which is $O(\sqrt{k} \log k + \polylog m)$ per arriving symbol.

\section{Proof of Lemma~\ref{lemma:validation} - A data structure for $k$-mismatch queries }\label{sec:ktime}
In this section we give the proof of Lemma~\ref{lemma:validation} which explains how we can maintain a small $k$-mismatch data structure that can be updated very quickly when a text symbol arrives but only computes an output at an alignment where a $k$-mismatch query is performed. The updates take $O(\polylog m)$ time and the queries take $O(k \polylog m)$ time.

\paragraph{The pattern and text partitioning.} The dynamic data structure we present here uses a simple, cyclic partitioning of the pattern and streaming text. The same partitioning will also be used in Sections~\ref{sec:approx} and~\ref{sec:smallperiod}. For an integer $q$ we can partition the pattern $P$ as follows: For each $r \in [0,q-1]$, the subpattern $P^{q, r} = P[r]P[q+ r]P[2q+r]\ldots P[\lfloor (m-r-1)/q \rfloor \cdot q + r]$. That is $P^{q, r}$ contains exactly the positions of $P$ that have remainder $r$ modulo $q$. The text stream can be partitioned into $r$ substreams analogously, i.e.\@ $T^{q, r} = T[r]T[q+ r]T[2q+r]\ldots\,$ for each $r \in [0,q-1]$. 

When $T[i]$ arrives in the text stream we refer to the alignment of $P$ and $T[i-m+1,i]$ as the \emph{current alignment}. There is also a natural notion of the \emph{current alignment} of subpattern $P^{q,r}$ with exactly one substream $T^{q,r'}$ for some $r' \in [0,q-1]$. Consider the positions in $P$ which correspond to positions in $P^{q,r}$. These positions in $P$ are aligned with $|P^{q,r}|$ positions in $T[i-m+1,i]$ which in turn all occur in some unique $T^{q,r'}$. In fact they exactly form the latest $|P^{q,r}|$ length suffix of the substream $T^{q,r'}$. We will refer to this alignment as \emph{the} current alignment of $P^{q,r}$ without explicitly referencing $T^{q,r'}$.

\paragraph{A randomised reduction to $1$-mismatch queries.}  We can assume that $m \ge \frac{34 k}{\delta} \log^2 m$. Otherwise, we can use $O(m)$ space and still satisfy the conditions for Lemma~\ref{lemma:validation}. In this case we maintain a data structure, as described in~\cite{CS:2010} which allows us to perform Longest Common Prefixes calls between the pattern and the latest $m$-length suffix of the streaming text, each taking constant time. We can see that at most $(k+1)$ Longest Common Prefixes calls are needed to answer a $k$-mismatch query and the update time per arriving symbol is $O(\log{m})$.  

We begin by giving a reduction to the $1$-mismatch problem. The reduction and the algorithm from Section~\ref{sec:approx} will use the following technical lemma.

\begin{lemma}\label{lm:isolated}
If $p_1, p_2$ are two distinct integers in $[1,m]$ and $q$ is a random prime number in the interval $[\frac{k}{\delta} \log^2 m, \frac{34 k}{\delta} \log^2 m]$ where $\frac{1}{6k} < \delta \leq 1$, then $Pr[p_1 = p_2 \bmod q] \leq \frac{\delta}{32 k}$. It is always assumed, unless otherwise stated, that ``$\log$" means $\log_2$.
\end{lemma}
\begin{proof}
We have $\frac{34 k}{\delta} \log^2 m > 17$. Applying Corollary 1 from~\cite{RS:1962} we obtain that the number of primes in the interval $[\frac{k}{\delta} \log^2 m, \frac{34 k}{\delta} \log^2 m]$  is at least

$$\frac{\frac{(34-2) \cdot k}{\delta} \log^2m}{\log{(\frac{34 k}{\delta} \log^2m)}} \ge \frac{\frac{32 k}{\delta} \log^2m}{\log m} \ge \frac{32 k}{\delta} \log m$$

If $p_1 = p_2 \bmod q$, then $q$ is a prime divisor of $|p_1-p_2|$. Observe that $|p_1 - p_2| \leq m-1$ has at most $\log m$ distinct prime divisors. Consequently, the probability that $q$ is one of these divisors is at most $\frac{\log m}{(32 k / \delta) \log m} = \frac{\delta}{32 k}$.
\end{proof}

We set $\delta$ to $1$ and pick $\log m$ primes independently and uniformly at random from $[\frac{k}{\delta} \log^2 m, \frac{34 k}{\delta} \log^2 m]$. These are denoted $q_1, q_2, \ldots, q_{\log m}$. Each $q_j$ gives a partitioning of $P$ into $q_j$ subpatterns $P^{q_j,r}$, and $T$ into~$q_j$ substreams  $T^{q_j,r}$, as described above. 

At the current alignment, that is the alignment of $P$ and $T[i-m+1,i]$, we say that a position in $P$ where a mismatch occurs is \emph{isolated} under $q_j$ if the current alignment of some subpattern $P^{q_j,r}$ containing that position has exactly one mismatch. We define $\Ih$ to be the number of positions in $P$ that are isolated mismatches between $P$ and $T[i-m+1,i]$ under at least one $q_j$. In Lemma~\ref{lem:IHam} below we demonstrate that if the latest Hamming distance is small then it equals $\Ih$ with high probability.

\begin{lemma}\label{lem:IHam}
If $\Hs{P}{T}{i} \le 2k$, then $\Hs{P}{T}{i} =\Ih$ with probability at least $1 -  \frac{1}{m^2}$.
\end{lemma}
\begin{proof}
$\Ham(P, T) [i] =\Ih$ if and only if each mismatch is isolated under $q_j$ for at least one $j$. Let $\Mh = \{x_1, x_2, \ldots, x_{|\Mh|}\}$ be the set of mismatches in the current alignment of $P$ and $T$. Suppose that a mismatch~$x_i$ is not isolated under $q_j$. It follows that $x_i = x_{i'} \bmod q_j$ for some $i' \neq i$. By Lemma~\ref{lm:isolated}, the probability of this event is at most $1 / 32 k$. Applying the union bound, we obtain that $x_i$ that is not isolated under $q_j$ with probability at most $1/16$. Therefore, as the primes are picked independently, a mismatch $x_i$ is not isolated under $q_j$ for all $j$ with probability at most $(1 / 16)^{\log m} = 1 / m^4$. Applying the union bound, we finally obtain that the probability of $\Ham(P, T)[i] \neq \Ih$ is at most $2k/m^4 \le 1 / m^2$.
\end{proof}

We will answer a $k$-mismatch query at alignment $i$ by computing $\Ih$. To allow us to compute $\Ih$, we will maintain a number of data structures that can answer \emph{$1$-mismatch queries} on the subpatterns. Given a pair $(q_j,r)$, a $1$-mismatch query determines whether at the current alignment of $P^{q_j,r}$ there is exactly one mismatch and if so, returns its location. By Lemma~\ref{lem:1mismatch} below, we can answer a $1$-mismatch query in $O(\polylog m)$ time.

\begin{lemma}\label{lem:1mismatch}
Given a pair $(q_j,r)$, a $1$-mismatch query on the current alignment of $P^{q_j, r}$ can be answered in $O(\polylog m)$ time. The required data structures use $O(k^2 \polylog m)$ total space and maintaining them takes $O(\polylog m)$ time when a stream update occurs.
\end{lemma}

We defer discussion of our method for answering  $1$-mismatch queries until after we explain how we use them to compute $\Ih$:  First, we perform $O(k \polylog m)$ $1$-mismatch queries to find the set containing every $(q_j,r)$ such that subpattern $P^{q_j,r}$ has exactly one mismatch. Second, we look through every $(q_j,r)$ in the set and use the position of the mismatch in  $P^{q_j, r}$ to determine the corresponding mismatching position in~$P$.  This set of mismatching positions is very likely to contain many duplicates because each position in $P$ occurs in exactly one $P^{q_j,r}$  for each $q_j$. Therefore, the third step is to remove any duplicates to recover $\Ih$. Finally we return  $\Ih$ as the answer to the $k$-mismatch query, unless  $\Ih > k$, in which case we return ``No''.

The total space is $O(k^2 \polylog m)$ and the update time is $O(\polylog m)$ both of which are dominated by the space and maintenance time of the data structures required to support $1$-mismatch queries. The time complexity for a $k$-mismatch query is therefore $O(k \polylog m)$ and is dominated by the time taken to perform $O(k \polylog m)$  $1$-mismatch queries, each taking $O(\polylog m)$ time.

\paragraph{Proof of Lemma~\ref{lem:1mismatch}.}  
We conclude this section by explaining our method for answering $1$-mismatch queries which is based on a reduction to streaming dictionary matching. Given a set of patterns~$D$, called a dictionary, the streaming dictionary matching problem is to find any occurrences of patterns in the dictionary in a text stream as they occur. We will use a recent streaming dictionary matching algorithm~\cite{CAPSS:2015} which is randomised and uses $O(|D| \log m)$ space and takes $O(\log \log m)$ time to process a stream update~--- i.e.\@ arrival of a new symbol of $T$. 

The dictionary that we build is based on a second level of partitioning of the subpatterns using the same partitioning scheme but with smaller values of $q$. For each (first-level) subpattern $P^{q_j,r}$ there is a set of $O(\log^2 m)$ second-level subpatterns which we denote by $\mathcal{P}_2^{q_j,r}$. 
From Theorem 1 in~\cite{RS:1962} it follows that there are at least $\log m / \log \log m$ primes in an interval~$[\log m, 3 \log m]$ and consequently the product of all primes in this interval is at least  $(\log m)^{\alpha} = m$. For each prime number $p \in [\log m, 3 \log m]$ there is a second-level subpattern $P^{q',r'} \in \mathcal{P}_2^{q_j,r}$ where $q'=(q_j \cdot p)$ and $r'=(q_j\cdot s)+r$. We define the dictionary $D = \bigcup_{q_j,r}  \mathcal{P}_2^{q_j,r}$ containing all $O(k \polylog m)$ second-level subpatterns. 

Each substream $T^{q_j,r}$ is partitioned into second-level substreams in an analogous manner. We run the streaming dictionary matching algorithm~\cite{CAPSS:2015} with dictionary $D$ on each second-level substream. Maintaining these streaming dictionary matching algorithms takes $O(\polylog m)$ time each time an update occurs. This is because each arriving $T[i]$ only occurs in $O(\log m)$ second-level substreams. For each substream we use $O(k \polylog m)$ space. As there are $O(k \polylog m)$ substreams this is $O(k^2 \polylog m)$ space in total.

Let us now show that a subpattern $P^{q_j,r}$ contains an isolated mismatch if and only if for each prime there exists exactly one second-level subpattern that does not match. Indeed, if $P^{q_j,r}$ contains an isolated mismatch then the second half of the statement obviously holds. Assume now that for each prime there exists exactly one second-level subpattern that does not match and that there are at least two mismatches at positions $1 \le x < y \le |P^{q_j, r}| < m$ in the current alignment of $P^{q_j,r}$.  For all $j$ the remainders of $x,y$ modulo $q_j$ are defined by the index of the second-level subpattern they belong to (i.e. the unique subpattern that does not match) and therefore are equal. As the product of the primes $q_j$ is at least $m$, by the Chinese Remainder Theorem we have $x = y$, a contradiction.

Therefore, to answer a 1-mismatch query on $P^{q_j,r}$ it suffices to determine which of the second-level subpatterns in $\mathcal{P}_2^{q_j,r}$ do not match, or, equivalently, match exactly at the latest alignment. With the help of the dictionary pattern matching algorithm we can find all second-level subpatterns $P^{q_i, r}$ that do not match in $O(\polylog m)$ time. If for each prime there is exactly one second-level subpattern that does not match, we can find the position of the mismatch in $P^{q_i, r}$ in $O(\polylog m)$ time as explained above.
\qed

\section{Proof of Theorem~\ref{thm:approximate} - A small space $(1+\eps)$-approximation}\label{sec:approx}
In this section we give our $(1+\eps)$-approximation for the streaming $k$-mismatch problem. If $\eps < 1/(2k)$, we can just run the $(1+1/(2k))$-approximate algorithm. This only improves the time and space, but does not change the output as the $(1+1/(2k))$-approximate algorithm exactly solves the $k$-mismatch problem and therefore by the definition gives a $(1+\eps)$-approximation. Below we assume $\eps \ge 1 / (2k)$. We will also assume that $m \ge \frac{34 k}{\delta} \log^2 m$, otherwise $O(m/\epsilon^2)$ space will satisfy the conditions for Theorem~\ref{thm:approximate} and we can simply apply the online version of Karloff's $(1+\epsilon)$-approximate algorithm~\cite{CEPP:2011}. 

Our algorithm, \approxAlg, will use the same partitioning of $P$ and $T$ into subpatterns $P^{q,r}$  and substreams $T^{q,r}$ as in Section~\ref{sec:ktime}. As before we will perform this partitioning for $O(\log m)$ values of $q$. However in contrast to Section~\ref{sec:ktime} the range from which the primes are chosen will also depend on $\eps$. Specifically,  ${q_1, q_2, \ldots, q_{\log m}}$ are picked independently and uniformly at random from the primes in the range $[\frac{k}{\delta} \log^2 m, \frac{34 k}{\delta} \log^2 m]$ where we set $\delta = \frac{\eps}{3}$. The subpatterns and substreams for $q_j$ then are given by $P^{q_j,r}$ and $T^{q_j,r}$ for each $r \in [0,q_j-1]$.

In Section~\ref{sec:ktime} we saw that for an arbitrary text substring $T[i-m+1,i]$ we can find the Hamming distance between $T[i-m+1,i]$ and $P$ (if it is small) by finding every subpattern $P^{q_j,r}$ that has exactly one mismatch. We will now see that to approximate the Hamming distance it suffices to count the number of subpatterns $P^{q_j,r}$ that do not match exactly. For some alignment $i$, let $\muij$ denote the number of subpatterns  $P^{q_j,r}$ that do not match exactly and let $\mui = \max_j \muij$. Lemma~\ref{lm:0_2k_approx} tells us that if the Hamming distance is small then $\mui$ is a good approximation of the true Hamming distance. As intuition for the proof techniques, first observe that $\muij$ is always upper-bounded by the true Hamming distance. The value of $\muij$  underestimates the Hamming distance whenever two mismatches in $P$ belong to the same subpattern $P^{q_j,r}$. Fortunately when the Hamming distance is relatively small, it is likely that for at least one prime $q_j$, the effect of these collisions will be small. Lemma~\ref{lm:>2k} shows that if $\Hs{P}{T}{i}$ is big, then $\mui$ is big with high probability. We will consider $\delta$ to be an arbitrary value between $1/(6k)$ and $1/3$.   

\begin{lemma}\label{lm:0_2k_approx}
 If $\Hs{P}{T}{i} \leq 2k$, then for all  $(1-\delta) \cdot \Hs{P}{T}{i} \leq \mui \leq \Hs{P}{T}{i} $ with probability at least $1 - \frac{1}{4m^2}$.
\end{lemma}
\begin{proof}
By definition, $\mui \leq \Hs{P}{T}{i} $ with probability $1$. 
Recall that $\mui = \max \muij$, where $\muij$ is the number of subpatterns $P^{q_j,r}$ that do not match. The number of such subpatterns is at least the number $\Ihj$ of mismatches isolated under $q_j$. Consequently, $\Ihj \leq (1-\delta) \cdot \Hs{P}{T}{i}$ for all $j$. It implies that the number $\bar{\Ih}_j$ of mismatches that are not isolated under $q_j$ is at least $\delta \cdot \Hs{P}{T}{i}$. On the other hand, $\mathrm{E} [\bar{\Ih}_j] \leq \frac{\delta}{16} \cdot \Hs{P}{T}{i}$ by Lemma~\ref{lm:isolated}. By Markov's inequality, the probability of $\bar{\Ih}_j \ge \delta \cdot \Hs{P}{T}{i}$ is at most $1 / 16$. As it holds for all $j$, the probability of $\mui \leq (1-\delta) \cdot \Hs{P}{T}{i}$ is at most $(1/16)^{\log m} < \frac{1}{4m^2}$.
\end{proof}

We now show that the Hamming distance is big, then $\mui$ is big with high probability.

\begin{lemma}\label{lm:>2k}
 If $\Hs{P}{T}{i} > 2k$ then $\mui > (1+\delta) \cdot k$ with probability at least $1 - \frac{1}{4m^2}$.
\end{lemma}
\begin{proof}
Suppose that $\Hs{P}{T}{i} > 2k$ and choose a subset $\Mh$ of any $2k$ mismatches between $P$ and $T[i-m+1,i]$. Remember that $\mui$ is the maximum number of subpatterns that do not match in a partition for the current alignment. We say that a mismatch $x$ is $\Mh$-isolated under $q_j$ if it is the only mismatch from $\Mh$ that occurs in the current alignment of some subpattern $P^{q_j,r}$. If $\mui \leq (1+\delta) \cdot k \leq \frac{5}{4} k$, then for all $j$ there are at most $\frac{5}{4} k$ subpatterns that do not match, and consequently there are at most $\frac{5}{4} k$ mismatches that are $\Mh$-isolated under $q_j$.

Assume that each mismatch $x \in \Mh$ is $\Mh$-isolated for more than $\frac{5}{8} \log m$ of the chosen primes. By summing over all mismatches in $\Mh$, we have that $\sum_j \muij > \frac{5}{4} k\log m$, a contradiction. Consequently, there is at least one mismatch $x \in \Mh$ that is not $\Mh$-isolated for at least $\frac{3}{8} \log m$ of the primes. 

By Lemma~\ref{lm:isolated} and the union bound the probability that a mismatch $x$ is not $\Mh$-isolated under $q_j$ is at most~$\delta/16$. So, the probability of $\Hs{P}{T}{i} > 2k$ is at most $(\delta / 16)^{\frac{3}{8} \log m} \leq \frac{1}{4 m^2}$. 
\end{proof}

As alluded to in Section~\ref{sec:overview}, algorithm \approxAlg performs two main phases. The first phase creates a set of $2\log m$ length-reduced versions of the pattern during preprocessing and then performs a series of transformations on the text as it arrives. There are two reduced patterns and two transformed texts for each of the $O(\log m)$ values of $q_j$. The second phase then approximates the Hamming distance between each of the reduced length patterns and the transformed texts. We will see that when combined these Hamming distances are a good approximation of $\mui$ which is in turn a good approximation of the true Hamming distance.

\paragraph{First phase.}  During the first phase, for each $q_j$ we perform a length reduction on $P$ by constructing two new patterns, $\phi_1^{q_j}$ and $\phi_2^{q_j}$, each of length $O(\frac{k}{\delta} \log^2 m)$. To this end, we first compute an identifier\footnote{For example, Karp-Rabin fingerprints~\cite{KR:1987} meet these requirements.}, denoted $\phi(P^{q_j,r})$, for each subpattern $P^{q_j,r}$ such that $\phi(P^{q_j,r})$ has $O(\log m)$ bits and with high probability $\phi(P^{q_j,r})=\phi(P^{q_{j'},r'})$ if and only if $P^{q_j,r}=P^{q_{j'},r'}$. For each $q_j$, either all the subpatterns have the same length or there exists an $s_j$ such that the  subpatterns $P^{q_j,0}, \ldots, P^{q_j,q_j-s_j-1}$  have equal lengths and  the subpatterns $P^{q_j,q_j-s_j}, \ldots, P^{q_j,q_j-1}$  which have length exactly one less.  If the subpatterns do have two different lengths, the two new patterns for prime $q_j$ are then given by $\phi_1^{q_j} = \phi(P^{q_j,0}) \ldots \phi(P^{q_j,q_j-s_j-1})$ and $\phi_2^{q_j} = \phi(P^{q_j,q_j-s_j}) \ldots \phi(P^{q_j,q_j-1})$. We will proceed assuming that not all the subpatterns have the same length as if they do we can simply omit the parts of the algorithm that would otherwise use the second pattern. 

We transform the text as it arrives to form two new streams, $C_1^{q_j}$ and $C_2^{q_j}$ for each $q_j$. To produce these new streams, for each substream $T^{q_j,r}$ we run two instances of a dictionary matching algorithm~\cite{CAPSS:2015}, one on dictionary $D_1 = \{P^{q_j,0}, \ldots, P^{q_j,q_j-s_j-1}\}$ and one on $D_2 = \{P^{q_j,q_j-s_j}, \ldots, P^{q_j,q_j-1}\}$. For the latest alignment in the substream $T^{q_j,r}$, each dictionary matching instance returns the identifier of a subpattern from its dictionary ($D_1$ or $D_2$) that currently matches (if there is one)\footnote{The streaming dictionary matching algorithm from~\cite{CAPSS:2015} can easily be modified to return such an identifier.}. Both instances use $O(q_j \log m)$ space and $O(\log \log m)$ time per position and are correct with high probability.

We use the output of the dictionary matching to form the streams, $C_1^{q_j}$ and $C_2^{q_j}$, for each $q_j$. When a new symbol in $T$ arrives, we will append one symbol to $C_1^{q_j}$ and one to $C_2^{q_j}$. The arrival of a new symbol in $T$ corresponds to a new symbol in one substream $T^{q_j,r}$ for each $q_j$. If we find a new match of a pattern from $D_1$ in  $T^{q_j,r}$  we append its identifier to $C_1^{q_j}$. Otherwise, we
append \$ to $C_1^{q_j}$. Analogously for $D_2$, we find a match of a pattern from $D_2$, we append its identifier to $C_2^{q_j}$, and otherwise we append \$. This allows us to compute $\muij$ at alignment $i$ as formalised by the following fact.

\begin{fact}\label{fct:ep}
For any alignment $i$ and $q_j$, we have that $ \muij = \Hs{\phi^{q_j}_1}{C_1^{q_j}}{i-s_j}  + \Hs{\phi^{q_j}_2}{C_2^{q_j}}{i}$.
\end{fact}
\begin{proof}
By definition, $\Hs{\phi^{q_j}_1}{C_1^{q_j}}{i-s_j}$ equals the number of subpatterns from $P^{q_j,0},  \dots, P^{q_j,q_j-s_j-1}$ that do not match at the current alignment, while $\Hs{\phi^{q_j}_2}{C_2^{q_j}}{i}$ equals the number of subpatterns among $P^{q_j,q_j-s_j}, \dots, P^{q_j,q_j-1}$ that do not match.
\end{proof}

\paragraph{Second phase.} The second phase approximates the values of $\Hs{\phi^{q_j}_1}{C_1^{q_j}}{i-s_j}$ and $\Hs{\phi^{q_j}_2}{C_2^{q_j}}{i}$ for each $q_j$ as the stream arrives. We compute these approximate Hamming distances  using an online variant~\cite{CEPP:2011} of Karloff's $(1+\delta)$-approximate pattern matching algorithm~\cite{Karloff:1993}. Karloff's algorithm requires $\delta$ to be bigger than the reciprocal of the pattern's length. This condition is satisfied as 

$$\delta \ge \frac{1}{6k} \ge \frac{1}{3k \log^2 m} \ge \frac{1}{\frac{k}{\delta} \log^2m} \ge \max\left( \frac{1}{|\phi^{q_j}_1|}, \frac{1}{|\phi^{q_j}_2|} \right)$$

The algorithm takes $O(\frac{k}{\delta^3} \log^4 m)$ space and $O(\frac{\log^4 {m}}{\delta^2})$ time per output. We run two instances of the algorithm for each $q_j$, one on the stream $C_1^{q_j}$ and the pattern $\phi^{q_j}_1$, and other  on stream $C_2^{q_j}$ and pattern $\phi^{q_j}_2$. For the first algorithm, we store the last $s_j \leq q_j$ outputs in a cyclic buffer. We can then compute $\wmuij$, the sum of the approximate values of $\Hs{\phi^{q_j}_1}{C_1^{q_j}}{i-s_j}$ and $\Hs{\phi^{q_j}_2}{C_2^{q_j}}{i}$ in $O(1)$ time per output. 

The maximum of the $\wmuij$ outputs over all $j$ is an integer $\wmui \in [\mui, (1+\delta) \cdot \mui]$, which can be computed in $O(\log m)$ time per position. The algorithm returns ``No'' if $\widetilde{\mui} > (1+\delta) \cdot k$ and ${\widetilde{\mui}}/{(1-\delta)}$ otherwise. The claim of correctness is given in Lemma~\ref{lm:ep-correctness}.

\begin{lemma}\label{lm:ep-correctness}
For all $\frac{1}{2k} < \eps \leq \frac{1}{2}$, if $\widetilde{\mui} > (1+\frac{\eps}{3}) \cdot k$, then $\Hs{P}{T}{i} > k$; otherwise, $\widetilde{\mui} /{(1-\frac{\eps}{3})}$ is a $(1+\eps)$-approximation of $\Hs{P}{T}{i}$. The error probability is at most $\frac{1}{m^2}$. 
\end{lemma}
\begin{proof}
We use Karp-Rabin fingerprints~\cite{KR:1987} as identifiers of the subpatterns. The probability that identifiers of two equal-length subpatterns are equal can be made as small as $1/n^3$ by choosing a sufficiently large prime. It implies that the probability of computing $\widetilde{\mui}$ incorrectly is at most $\frac{(34 k / \delta) \log^2 m} {n^3} \le 1 / (4m^2)$.  Assume that $\widetilde{\mui}$ is computed correctly. If $\widetilde{\mui} > (1+\delta) \cdot k$, then $\Hs{P}{T}{i} \ge \mui \ge \widetilde{\mui}/(1+\delta) > k$. Otherwise, $\mui \leq \widetilde{\mui} \leq (1+\delta) \cdot k$, and from Lemma~\ref{lm:0_2k_approx} we obtain that $\Hs{P}{T}{i} \leq 2k$ with probability at least $1- 1/(4m^2)$. Finally, Lemma~\ref{lm:0_2k_approx} also implies that $\Hs{P}{T}{i} \leq \mui / (1-\delta) \leq \widetilde{\mui} / (1-\delta)$ and $\widetilde{\mui} / (1-\delta) \leq \frac{1+\delta}{1-\delta} \cdot \mui \leq (1+\eps) \cdot \mui \leq (1+\eps) \cdot \Hs{P}{T}{i}$ with probability at least $1-1/ (4m^2)$. The output is the integer $\lfloor \widetilde{\mui} / (1-\delta) \rfloor \leq \widetilde{\mui} / (1-\delta) \leq (1+\eps) \cdot \Hs{P}{T}{i}$. As $\mui / (1-\delta) \geq \Hs{P}{T}{i}$ and $\Hs{P}{T}{i}$ is an integer we have that $\lfloor \widetilde{\mui} / (1-\delta) \rfloor \geq \Hs{P}{T}{i}$. The claim follows. 
\end{proof}

\paragraph{Time and space complexities.}  It suffices to estimate the overall time and space complexities for the case where $\eps \ge 1/(2k)$ as for the smaller values of $\eps$ we run a $(1+1/(2k))$-approximate algorithm. For one prime and one substream, the dictionary pattern matching algorithm uses $O\bigl( (k / \delta) \log^3 m \bigr)$ space as the dictionary will contain $O\bigl( (k / \delta) \log^2 m \bigr)$ subpatterns. In total, all the dictionary pattern matching algorithms combined use $O\bigl( (k^2 / \delta^2) \log^6 m \bigr)=O\bigl( (k^2 / \eps^2) \log^6 m)$ space as we have $O(\log{m})$ primes for each of the $O(\bigl(k / \delta) \log^2 m \bigr)$ substreams. We also require $O\bigl( (k / \eps^3) \log^5 m \bigr)$ space  to run all $O(\log{m})$ copies of the online version of Karloff's $(1+\delta)$-approximation algorithm. This is because each subpattern is of length $O((k/\epsilon) \log^2{m})$  (recall that $\delta = \eps / 3$). Despite this the overall space complexity is not affected by running Karloff's algorithm. This is because if $\eps > 1/2k$ then the space is dominated by $O\bigl( (k^2 / \eps^2) \log^6 m)$. 

Each symbol of $T$ is added to only one of the substreams $T^{q_j, r}$ for each $j$. For each of them we update the dictionary matching algorithms, which takes $O(\log m \log \log m)$ time. Next, for each of the $O(\log{m})$ updated streams we give one output of the online version of Karloff's algorithm, which takes $O(\log^5 {m} / \delta^2) = O(\log^5 {m} / \eps^2)$ time in total. This completes the proof of Theorem~\ref{thm:approximate}.

\section{Proof of Lemma~\ref{lemma:smallperiod} - The small approximate period case}\label{sec:smallperiod}
We now give a proof of Lemma~\ref{lemma:smallperiod} which states that if the $3k$-period of $P$ is smaller than $k$, then the $k$-mismatch pattern matching problem can be solved in $O(k^2)$ space and $O(nk^2\log{k}/m +n)$ time. The discussion follows with reference to the steps of Algorithm~\ref{alg:det-smallperiod} which is given in Section~\ref{sec:overview}. 

Our algorithm utilises a simple variant of run length encoding. We will use this encoding to reduce the $k$-mismatch problem to a total of $O(k^2)$ small instances of the run length encoded Hamming distance problem. Each instance will process a pattern/text pair each containing $O(k)$ runs. By using a streaming variant of an existing run length encoded Hamming distance algorithm, we will be able to output the Hamming distances for each of these instances in a compressed format in a total of $O(k^2 \log k)$ time. The original Hamming distances can then be recovered in a streaming fashion by summing the outputs of the run length encoded instances. 

\paragraph{Run length encoding using the $3k$-period.} We begin by describing the variant of run length encoding that we will use and argue that all the information about the pattern and text that we need to answer $k$-mismatch queries can be encoded in $O(k)$ space.  Let $\ell \leq k$ be the $3k$-period of $P$. We partition the pattern and the text as described in Section~\ref{sec:ktime} except that instead of choosing a random prime, we use the fixed value $\ell$ instead. Recall that for an arbitrary string $S$, the partition $S^{\ell,r}$ is defined to be equal $S[r]S[\ell+r]S[2\ell+r]\ldots$ up until the end of $S$. As $\ell$ is fixed for this section, we will shorten the notation $S^{\ell,r}$ to $S^r$ instead. The $\ell$-run length
encoding of a string $S$ is defined as the ordered set of all $S^{r}$,
each stored in run length encoded form, where $r \in [0,\ell-1]$. We denote by
$\runs(S^{r})$ the number of runs in $S^{r}$. The size of the
encoding, denoted $\allruns(S)$ is $\sum_{r=0}^{\ell-1} \runs(S^{r})$. We begin with
an example of the encoding. The whitespace in $P$ in the example has only been included for visual clarity.


\begin{example}
Let $P= aab\  aab\ aab\ aab\  aab\ aab\ aac$ and $k=4$. The $3k$-period of $P$ is $\ell=3$. We then have that,  $P^{0}={aaaaaaa}, P^{1}={aaaaaaa}, P^2 = {bbbbbbc}$. The $\ell$-run length encoding of $P$ is: the run length encoding $(a,7)$ of $P^0$, the run length encoding $(a,7)$ of $P^1$, and the run length encoding $(b,6)(c,1)$ of $P^2$. The size of the encoding , $\allruns(P)=1+1+2=4$.
\end{example}

\noindent Our first observation is that for a pattern with small approximate period, its $\ell$-run length encoding is also small. Intuitively this is because a pattern with small approximate period \emph{almost} repeats every $\ell$ symbols. 

\begin{lemma}\label{lem:RLEsubstring}
If $P$ has $3k$-period at most $k$ then $\allruns(P) \leq 4k$. 
\end{lemma}
\begin{proof}
We have that $\Ham(P[\ell,m-1],P[0,m-1-\ell])\leq
3k$. Let $h=\Ham(P[\ell,m-1],P[0,m-1-\ell])$ and let
$\mathcal{I}=\{i_1,i_2,\ldots i_h\}$ be the set of locations of the
mismatches in $P[0,m-1-\ell]$. For all $i \in [\ell,m-1]\setminus \mathcal{I}$ we have
that $P[i-\ell]=P[i]$. Furthermore let $\mathcal{I}_r$ be the subset of
$\mathcal{I}$ containing indices $\{ i \in \mathcal{I} ~|~ i =
r \bmod \ell \}$. Observe that for $r,r' \in [0,\ell-1]$ with $r\neq
r'$, we have that $\mathcal{I}_r$ and $\mathcal{I}_{r'}$ are
disjoint. Recall that $P[i-\ell]=P[i]$ for all $i \in [\ell,m-1]\setminus \mathcal{I}$. If we rephrase this in terms of $P^{r}$, we have that $P^{r}[q-1]=P^{r}[q]$ if $(q\ell + r) \in [\ell,m-1]\setminus
\mathcal{I}_r$. Since the number of runs in $P^{r}$ is equal to the number of non-equal neighbouring symbols plus one, the number of runs in $P^{r}$ is at most $|\mathcal{I}_r|+1$. By summing over all $r$, we have that $\allruns(P) \leq 3k + \ell \leq 4k$.
\end{proof}

The second observation is that there is a substring of $T$ which we call $\Ts$  which compresses well and contains every alignment with at most $k$ mismatches with the pattern. Intuitively this substring compresses well because it is very similar to the pattern, which in turn compresses well. Let us define $T_L$ to be the longest suffix of $T[0,m-1]$ for which $\allruns(T_L) \leq 5k$ and $T_R$ to be the longest prefix of $T[m,2m-1]$ for which $\allruns(T_R) \leq 5k$. We define $\Ts = T_LT_R$. It follows directly that $\allruns(T_\star) \leq 10k$. 

\begin{lemma}
$\Ts$ completely contains every $T[i-m+1,i]$ such that $\Hs{P}{T}{i} \leq k$.
\end{lemma}
\begin{proof}
Let $i_L$ be the smallest integer such that $\Ham(P,T)[i_L+m-1] \leq k$ and let $i_R$ be the largest integer such that $\Ham(P,T)[i_R] \leq k$. Obviously, $T[i_L, i_R]$ completely contains every $T[i-m+1,i]$ such that $\Hs{P}{T}{i} \leq k$.

To show that $\Ts$ contains $T[i_L, i_R]$ it suffices to show that the run length encodings of $T[i_L,m-1]$ and $T[m, i_R]$ have size at most $5k$. To see that $\allruns(T[i_L, m-1]) \leq 5k$, consider alignment $i_L+m-1$. As $\Ham(P,T)[i_L+m-1] \leq k$ and $m-1 \le i_L+m-1$, we have that $P$ differs from $T[i_L,i_L+m-1]$ in at most $k$ positions. However, we have just shown
that $\allruns(P) \leq 4k$. Consider the run length encoding of $P^{r}$ and the encoding of $T^{r}$. If there is a run in the encoding of $T^{r}$ which ends at some $T[i_L+j]$ but there is no run ending at $P[j]$, then this must be the position of a mismatch. Therefore the number of these additional runs is at most $k$. Furthermore, we have that $P[j]$ is such that $j= r \bmod \ell$. Therefore the mismatch $P[j]$ cannot cause an additional run in any $T^{r'}$ with $r' \neq r$. We therefore have that by summing over all $r$, the total number of runs, $\allruns(T[i_L,i_L+m-1])$ is at most $\allruns(P)+k \leq 5k$. Finally we observe that the encoding of a prefix is no larger than the encoding of the original. That is, $\allruns(T[i_L,m-1]) \leq \allruns(T[i_L,i_L+m-1]) \leq 5k$. An
analogous argument allows us to prove that $\allruns(T[m,i_R]) \leq 5k$.
\end{proof}

\paragraph{Run length encoded Hamming distance.}
Before we explain the full algorithm in more detail, we first introduce the algorithm \rleAlg. The algorithm \rleAlg is a straightforward adaptation of the offline algorithm of Chen et al.~\cite{CHC:2010}, which computes Hamming distances between run length encoded text and pattern, to the streaming setting. 

We briefly explain the overall approach of Chen et al.'s algorithm~\cite{CHC:2010}. Consider a text $T'$ and a pattern $P'$ both in the run length encoded form. Let $D$ be an $m \times n$ matrix where $D[i,j]$ equals one if $P'[j]\neq T'[i]$ and equals zero otherwise. The Hamming distance between $P'$ and $T'[i-m+1,i]$ is exactly the sum of the entries along the $i$-th diagonal of $D$. The $i$-th diagonal is the one which intersects cells $D[i-m+1,0]$ and $D[i,m-1]$. The first observation that Chen et al. make is that the matrix $D$ can be composed into $O( \runs(P') \cdot \runs(T'))$ monochromatic rectangles. These rectangles are exactly given by dividing $D$ horizontally whenever $P'[j]\neq P'[j-1]$ and vertically whenever $T'[i] \neq T'[i-1]$. For $1 \le i \le |P'|$, they define $\HD[i]$ to be the difference between the Hamming distance at alignments $i$ and~$(i-1)$. Formally, 

$$\HD[i] =  \Hs{P'}{T'}{i}- \Hs{P'}{T'}{i-1}$$

Further they observe that if the $i$-th diagonal does not intersect any corners then $\HD[i]=\HD[i-1]$. In an offline setting, the values of $\HD[i]$ such that  $\HD[i] \neq \HD[i-1]$ (and hence the values of $\Hs{P'}{T'}i$) can be found by sorting these corners and processing them in the order that they intersect the $i$-th diagonal as $i$ increases. 

We begin by briefly explaining how the input and output have been adapted for our streaming setting. The \rleAlg algorithm consists of two alternating operations, $\update(i,\sigma)$ and $\query(i)$. The input to \rleAlg  is supplied via the $\update(i,\sigma)$ operation which informs algorithm \rleAlg that a new run starts at $T'[i]=\sigma$. Each $\update(i,\sigma)$ operation triggers $\query(i)$ operation.

Operation $\query(i)$ produces an output of the algorithm. $\query(i)$ returns three values: a pair $(\HD[i], i^*)$, where $i \le i^*$, and $\Hs{P'}{T'}{i}$. Next $\query$ operation will be called at next $\update$ operation or at $T[i^*]$, whichever comes first. It is guaranteed that if no \update occurs during $T'[i,i^*]$ then $\HD[i] = \HD[i+1] = \ldots = \HD[i^*-1]$.

We now explain how the operations $\update$ and $\query$ are supported. We maintain a diagonal line which moves from left to right as $\update$ and $\query$ operations occur. When either $\update(i,\sigma)$ or $\query(i)$ is performed, the diagonal line moves forward to the $i$-th diagonal. Any corners of rectangles in $D$ that are crossed by the movement of the line are processed in order. This is achieved using a priority queue containing currently unprocessed corners (sorted by the order that the corners intersect the $i$-th diagonal). As all points which are to the left of or are currently on the $i$-th diagonal have been processed by the end of  $\query(i)$, both $\HD[i]$ and $\Hs{P'}{T'}{i}$ can be outputted by following the approach of Chen et al. Following the discussion above, any $\update$ operation corresponds to a new vertical line in $D$. This introduces $O(\runs(P'))$ rectangles and hence $O(\runs(P'))$ new corners. These points are pushed into the priority queue when $\update$ operation occurs. Finally for any $\query(i)$ operation we also need to output~$i^*$, where $i^*\geq i$ is the smallest integer such that there is a corner currently in the priority queue which intersects diagonal $i^*$. We can find this value with the help of the priority queue. Observe that the number of distinct~$i^*$ outputted by the algorithm over all $\query(i)$ operations is upper-bounded by the number of corners which is $O( \runs(P')\cdot  \runs(T'))$. This property is required when we use the algorithm to limit the number of   $\query(i)$ operations required.
We now summarise the space and time complexities of the \rleAlg algorithm in Lemma~\ref{lem:RLEham}. 

\begin{lemma}\label{lem:RLEham}
Given a run length encoded pattern $P'$ and text $T'$, the algorithm \rleAlg solves the Hamming distance problem in $O(\runs(P'))$ space. The amortised time complexity of $\update$ or $\query$ operation is $O( \runs(P') \log (\runs(P')))$ or $O(\log (\runs(P')))$ respectively.  No preprocessing is needed.
\end{lemma}
\begin{proof}
The space complexity follows from Chen et al. who observe that the size of the priority queue is $O(\runs(P'))$ at any time. The whole of $P'$ can be stored in $O(\runs(P'))$ space. Only the latest symbol of~$T'$ is required. 

Recall that the time complexities are amortised over all $\update$ and $\query$ operations performed so far. The number of points inserted into the priority queue is $O(\runs(P'))$ per $\update$ performed. A cost of $O(\runs(P')  \log (\runs(P')))$  is charged to the \update which inserted them. This pays for processing them during any subsequent $\update$ or $\query$ operations. The amortised time complexity of $\update$ operation is therefore $O(\runs(P') \log (\runs(P')))$ because priority queue operations take  $O(\log (\runs(P')))$ time. Similarly, the  amortised time complexity of the $\query$ operation is $O(\log (\runs(P')))$.
\end{proof}

\paragraph{The $k$-mismatch algorithm.} 
We now give our full algorithm for the $k$-mismatch problem in the small approximate period case. Recall that in this section we assume that $|T|=2m$. The algorithm performs three phases, \emph{Setup}, \emph{Handover} and \emph{Output} depending on the value of $i$ when $T[i]$ arrives. The symbol $T[m-1]$ is processed by all three phases (in ascending order) and is the only symbol processed by the Handover phase.

\bigskip

\noindent {\bf Setup phase:}\ $(i \leq m-1)$. We maintain a modified $\ell$-run length encoding of the longest suffix $T_L$ of the current text $T[0,i]$ such that $\allruns(T_L) \leq 5k$ (see Lemma~\ref{lem:RLEadd}). More formally, we maintain for each $r \in [0,\ell-1]$ a linked list of tuples $(j, T[j])$, where $j$ are the starting positions of runs in $T_L^{s}$ for $s = {i_1 + r \bmod \ell}$. We also maintain the length of each list and the total length of all lists.

\medskip

\noindent {\bf Handover phase:}\ $(i = m-1)$. We compute the $\ell$-run length encoding of $T_L$ and then start $\ell^2$ instances of \rleAlg. For each $(r,s) \in [0,\ell-1]^2$, the instance denoted $\rleAlg{(r,s)}$ uses pattern $P^{r}$ and text $T_L^{s'}$, where $s'+m-|T_L| = s \bmod \ell$. A sequence of \update operations are performed immediately on $\rleAlg{(r,s)}$ to provide the whole of the run length encoding of $T_L^{s'}$ as text input. The \update operations are offset to account for the start of $T_L^{s'}$ within $T^{s}$. Specifically, for each $T_L^{s}[i'] \neq T_L^{s}[i'-1]$ we perform $\update(i' + \lfloor (m-s)/\ell \rfloor - |T_L^{s}|, T_L^{s}[i'])$. 

\medskip

\noindent {\bf Output phase:}\ $(i \geq m-1)$. We perform four steps:

\begin{enumerate}
\item First, we check whether $T[i]$ starts a new run in $T^{s}$ where $s = i \bmod \ell$. If so for each $r \in [0,\ell-1]$, we perform $\update(\lfloor i/\ell \rfloor, T[i])$ on instance $\rleAlg(r,s)$. Recall that every $\update(\lfloor i/\ell \rfloor, T[i])$ operation also triggers a $\query(\lfloor i/\ell \rfloor)$ operation.

\item Second, for each $r \in [0,\ell-1]$ we compute $\HD_{r,s}[\lfloor i/\ell \rfloor]$ - the value of $\HD[\lfloor i/\ell \rfloor]$ for instance $\rleAlg(r,s)$  where $s = i \bmod \ell$. To this end we determine the set of all $r \in [0,\ell-1]$ such that $i^*_{r,s} = \lfloor i/\ell \rfloor$. Here $i^*_{r,s}$ is the $i^*$ value outputted by the last \query operation performed on $\rleAlg(r,s)$. For every such $\rleAlg(r,s)$ we perform $\query(\lfloor i/\ell \rfloor)$ to compute $\HD_{r,s}[\lfloor i/\ell \rfloor]$ and then update $i^*_{r,s}$. For all other  $(r,s)$, we have that $\HD_{r,s}[\lfloor i/\ell \rfloor]=\HD_{r,s}[\lfloor i/\ell \rfloor -1]$.

\item Third, we check whether the total number of runs processed by all \rleAlg instances exceeds $8k$. If so, all \rleAlg instances are abandoned and we output ``No'' for this and every subsequent value of $i$ in $[m-1,2m-1]$.

 \item Finally, we compute the latest Hamming distance, $\Hs{P}{T}{i}$ from $\Hs{P}{T}{i-\ell}$ and the outputs of the $\rleAlg(r,s)$ using the equations from Lemma~\ref{lem:HamSumCor} and Lemma~\ref{lem:HamSum} as described below.
\end{enumerate}  

All steps of the algorithm are self-explanatory, except for the Setup phase and the fourth step of the Output phase, which we describe in details below. We start by giving a lemma that will allow us to compute $T_L$ (the Setup phase).

\begin{lemma}\label{lem:RLEadd}\label{lem:RLEremove}
Given the modified $\ell$-run length encoding of $S = T[i_1, i_2]$, the modified $\ell$-run length encoding of either $T[i_1+1,i_2]$ or $T[i_1,i_2+1]$ can be computed in $O(1)$ time.
\end{lemma}
\begin{proof}
To compute the encoding of $T[i_1+1,i_2]$, we go to the $(i_1\bmod \ell)$-th list. The first two tuples in this list define the length of the first run in $S^{(i_1\bmod \ell)}$. If it equals one, we delete the first tuple and then decrement the length of the list and the total length of the lists by one. Otherwise, we simply replace the first tuple by $(i_1+\ell, T[i_1+\ell])$.

To compute the encoding of $T[i_1,i_2+1]$, we go to the $((i_2+1)\bmod \ell)$-th list. The last tuple in the list defines whether $T[i_2+1]$ starts a new run in $S^{((i_2+1) \bmod \ell)}$. If it does, we add a new tuple $(i_2+1, T[i_2+1])$ to the list and increment the list's length and the total length by one. Otherwise, we do nothing.
\end{proof}

We now give two lemmas which combined will allow us to efficiently compute the final Hamming distances (the fourth step of the Output phase). Note that the \rleAlg instances collectively process the substring $\Ts$ as defined in Lemma~\ref{lem:RLEsubstring}. Let $\Ts = T[i'_L,i'_R]$. (Recall that $\Ts$ contains $T[i_L,i_R]$ but does not necessarily equal it). Remember that for any $i \not \in [i'_L + m-1,i'_R]$, we have that $\Hs{P}{T}{i} > k$. For the first $\ell$ alignments in $[i'_L+m-1, i'_R]$ we use Lemma~\ref{lem:HamSumCor} to calculate the output directly from the \rleAlg outputs. 

\begin{lemma}\label{lem:HamSumCor} For any $i  \in  [i'_L +m-1, i'_R]$, we have that

\[\Hs{P}{T}{i} = \sum_{r=0}^{\ell-1} \Hs{P^{r}}{T^{R(r,i)}}{Q(r,i)},\]

\noindent where $R(r,i) = (r +i-m+1) \bmod \ell$ and  $Q(r,i) = \left\lfloor\tfrac{r +i-m+1}{\ell}\right\rfloor + |P^r| -1$. 
\end{lemma}
\begin{proof}
In the alignment of $P$ and $T[i-m+1,i]$ we have that $P^r$ is aligned against $T[i-m+1+r]T[i-m+1+r+\ell]\dots T[i-m+1+r+\ell\cdot (|P^r|-1)]$. The claim follows.
\end{proof}

For the remaining alignments we use Lemma~\ref{lem:HamSum}. We will compute $\Hs{P}{T}{i}$ from $\Hs{P}{T}{i-\ell}$ and $\HD^\ell[i]$, where $\HD^\ell[i] = \sum_{r=0}^{\ell-1} \HD_{r,R(r,i)}{Q(r,i)}$. The value of $\HD^\ell[i]$ will in turn be computed from $\HD^\ell[i-\ell]$ by updating only the terms which have changed. We will argue below that these terms change very rarely.

\begin{lemma}\label{lem:HamSum}
$  \Hs{P}{T}{i} - \Hs{P}{T}{i-\ell} = \sum_{r=0}^{\ell-1} \HD_{r,R(r,i)} [Q(r,i)]$
\end{lemma}
\begin{proof}
First consider Lemma~\ref{lem:HamSumCor}  with $i$ substituted for $i-\ell$. We have that,

\[\Hs{P}{T}{i-\ell} = \sum_{r=0}^{\ell-1} \Hs{P^{r}}{T^{R(r,i-\ell)}}{Q(r,i-\ell)}  \]

It follows from the definitions of $R$ and $Q$ that $R(r,i-\ell)= R(r,i)$ and $Q(r,i-\ell)=Q(r,i)-1$. This therefore simplifies to 

\[\Hs{P}{T}{i-\ell} = \sum_{r=0}^{\ell-1} \Hs{P^{r}}{T^{R(r,i)}}{Q(r,i)-1} . \]

\noindent We therefore have that $\Hs{P}{T}{i} - \Hs{P}{T}{i-\ell}$ equals 

\[  \sum_{r=0}^{\ell-1} \left( \Hs{P^{r}}{T^{R(r,i)}}{Q(r,i)} - \Hs{P^{r}}{T^{R(r,i)}}{Q(r,i)-1} \right). \]

\noindent From the algorithm description it then follows that, \[ \HD_{r,R(r,i)}[Q(r,i)]  = \Hs{P^{r}}{T^{R(r,i)}}{Q(r,i)} - \Hs{P^{r}}{T^{R(r,i)}}{Q(r,i)-1}.\] The claim follows immediately via substitution.
\end{proof}

\paragraph{Space complexity.} We now establish that the space complexity of the $k$-mismatch pattern matching algorithm is $O(k^2)$ as stated in Lemma~\ref{lemma:smallperiod}. 
The space required to store $P$ in the $\ell$-run length encoded form as well as the suffix~$T_L$ is $O(k)$ by definition. To compute the latest Hamming distance we store the most recent $\ell$ Hamming distances as well as the last two outputs from each \query operation on each $\rleAlg$ instance. Only these \query outputs are required because $Q(r,i) \in [ \lfloor i/\ell \rfloor -1, \lfloor i/\ell \rfloor]$ as we show in Lemma~\ref{lem:Q}. 

\begin{lemma}\label{lem:Q}
$Q(r,i) \in \left[\left\lfloor i/\ell \rfloor -1 , \lfloor i/\ell \right\rfloor \right]$.
\end{lemma}
\begin{proof}
Finally we demonstrate the observation that $ Q(r,i) \in [\lfloor i/\ell \rfloor -1 , \lfloor i/\ell \rfloor  ]$. Substituting in the length of $P^r$ we have that $Q(r,i)$ equals $\left\lfloor\tfrac{r +i-m+1}{\ell}\right\rfloor + (\left\lfloor\tfrac{m-r-1}{\ell}\right\rfloor+1)-1$. Further,

\[   \left\lfloor \frac{i}{\ell}\right\rfloor -1 \leq  \left\lfloor \frac{r +i-m+1}{\ell}\right\rfloor + \left\lfloor \frac{m-r-1}{\ell}\right\rfloor \leq \left\lfloor \frac{i}{\ell}\right\rfloor \]
\end{proof} 
 
As there are $\ell^2$ different   $\rleAlg$ instances, this is $O(k^2)$ space. Finally we have to account for the working space of the $\rleAlg$ instances. For any fixed $s \in [0,\ell-1]$ the space used by all $\rleAlg(r,s)$ instances is $\sum_{r=0}^{\ell-1} \runs(P^r)= O(k)$, which is $O(k^2)$ space over all $s$. Therefore, the space complexity is $O(k^2)$ overall as claimed.

\paragraph{Time complexity.} Finally, we show that the time complexity of the $k$-mismatch pattern matching algorithm is $O(nk^2\log{k}/m +n)$. The time complexity of the Setup phase is $O(1)$ time per symbol, or $O(m)$ time overall, by Lemma~\ref{lem:RLEadd}. The Handover phase starts by computing the $\ell$-run length encoding of $T_L$ from the modified encoding maintained through the Setup phase, which can be done in $O(k)$ time. It then performs the initialising \update operations on the  $\rleAlg$ instances. The total time complexity for all operations on the $\rleAlg$ instances will be accounted for below. 

The Output phase is split into four steps. The first step is also dominated by the \update operations on  the $\rleAlg$ instances. The second step can be implemented so that the time complexity is dominated by the \query operations performed. In particular we need to avoid spending $O(\ell)$ time to check whether each  $r \in [0,\ell-1]$ has  $i^*_{r,s} = \lfloor i/\ell \rfloor$. For each $s$ we maintain a sorted linked list of the current values of each $i^*_{r,s}$. We can then find all $i^*_{r,s} = \lfloor i/\ell \rfloor$ in time proportional to the number of such $i^*_{r,s}$ which in turn is equal to the number of \query operations performed. The third step takes $O(1)$ time per symbol via a simple counter, i.e. $O(m)$ time in total. 

Finally, we discuss the fourth step of the Output phase. To compute the Hamming distances for $i \in [i'_L,i'_L+\ell-1]$, we apply Lemma~\ref{lem:HamSumCor}. This takes $O(\ell)$ time per symbol which is $O(\ell^2) = O(k^2)$ time in total. For the remaining Hamming distances we apply Lemma~\ref{lem:HamSum}. This would take $O(\ell)$ as well if we applied it directly. To avoid this, we compute the value of $\HD^\ell[i]$ from the value of $\HD^\ell[i-\ell]$ by determining which terms have changed and updating them.  

\begin{fact}\label{fact:frfq} 
$\HD^\ell[i] = \sum_{r=0}^{\ell-1} \HD_{r,R(r,i-\ell)}[Q(r,i-\ell)+1]$.
\end{fact}
\begin{proof}
From the definitions of $R$ and $Q$ we have that $R(r,i) =R(r,i-\ell)$ and $Q(r,i) = Q(r,i-\ell) +1$.  
\end{proof}

On the other hand, $\HD^\ell[i-\ell] = \sum_{r=0}^{\ell-1} \HD_{r,R(r,i-\ell)}[Q(r,i-\ell)]$ by definition. By storing the most recent~$\HD_{r,s}$ values for all $(r,s)$ (see Lemma~\ref{lem:Q}), it is straightforward to determine which terms have changed in time proportional to the number of terms that have changed. Furthermore, for $i_1 \neq i_2 \bmod \ell$ and $r \in [0,\ell-1]$, we have that $R(r,i_1) \neq R(r,i_2)$. Consequently, for any $(r,s, j)$, there is at most one value of $i$ such that  $\HD_{r,s}[j]$ appears as a term in the expression for $\HD^\ell[i]$. Therefore the total time complexity for step four is upper-bounded by the number of $(r,s,j)$ such that  $\HD_{r,s}(j) \neq \HD_{r,s}(j-1)$. This is in turn upper-bounded by the total number of \update and \query operations performed.

Remember that the total number of \update and \query operations performed by all instances of \rleAlg is at most $O(\runs(P) \cdot \runs (\Ts)) = O(k^2)$. Therefore, the total time complexity is $O(m+k^2)$ excluding the time taken to perform the \update and \query operations. It remains to give an upper bound on the total number of these operations for each $\rleAlg$. For a given $(r,s)$, the number of \update operations on $\rleAlg(r,s)$ is $O(\runs(T^s))$. 

The total time spent performing \update and \query operations on $\rleAlg(r,s)$ is therefore $O(\runs(P^{r}) \cdot \log( \runs(P^{r}) ) \cdot \runs(T'_{s}))$.  Summing over all  $\rleAlg$ instances, and simplifying, we have that  

\[\sum_{r,s} O(\runs(P^{r}) \cdot \runs(T^{s}) \cdot \log k) = O\left(\sum_{r}  \runs(P^{r}) \cdot \sum_{s}  \runs(T^{s}) \cdot \log k \right) = O(k^2 \log k).\]

Therefore the total time complexity of the entire algorithm is $O(m+k^2\log k)$. It is important for the deamortised algorithm we give in Theorem~\ref{thm:streaming} (which uses this algorithm as a black box) that if $m \ge 2k^2$ then for processing any $k^2$ consecutive text symbols we spend only $O(k^2\log k)$ time as the term $m$ in the time complexity comes from spending $O(1)$ time per symbol in the worst case.

\section{Acknowledgements}
We thank Hjalte Wedel Vildh{\o}j for pointing out a typo in our definition of small and large approximate period cases.   

\bibliographystyle{plain-fi}
\bibliography{longnames,bristol,bib-latest}
\end{document}